\documentclass[11pt]{article}
\usepackage{amsmath}
\usepackage{amsthm}
\usepackage{amssymb}
\usepackage{amscd}
\usepackage{relsize}
\usepackage{enumerate}
\usepackage{mathdots}
\usepackage[pdftex]{graphicx}

\newtheorem{theorem}{Theorem}

\newtheorem{pro}{Proposition}
\newtheorem{cor}{Corollary}

\date{}
\begin{document}

\title{\bf Filter Banks on Discrete  Abelian Groups}

\author{
{\bf A.~G. Garc\'{\i}a}\thanks{E-mail:\texttt{agarcia@math.uc3m.es}}, \,\,
{\bf M.~A. Hern\'andez-Medina}\thanks{E-mail:\texttt{miguelangel.hernandez.medina@upm.es}}
{\bf\,\, and \,\, G. P\'erez-Villal\'on}\thanks{E-mail:\texttt{gperez@euitt.upm.es}}
}
\maketitle
\begin{itemize}
\item[*] Departamento de Matem\'aticas, Universidad Carlos III de Madrid,
 Avda. de la Universidad 30, 28911 Legan\'es-Madrid, Spain.
\item[\dag\ddag] Departamento de Matem\'atica Aplicada a las Tecnolog\'{\i}as de la Informaci\'on y las Comunicaciones, E.T.S.I.T., U.P.M.,
 Avda. Complutense 30, 28040 Madrid, Spain.
\end{itemize}
\begin{abstract}
In this work we provide polyphase, modulation, and frame theoretical analyses of a filter bank on a  discrete abelian group. Thus,  multidimensional or cyclic filter banks as well as filter banks for signals  in $\ell^2(\mathbb{Z}^d\times \mathbb{Z}_s)$ or 
$\ell^2(\mathbb{Z}_r \times \mathbb{Z}_s)$ spaces are studied in a unified way. We obtain  perfect reconstruction conditions and the corresponding frame bounds.  
\end{abstract}

\noindent{\bf AMS}: 22B05; 42C15; 94A12.
\section{Introduction}

The aim of this paper is to provide a filter bank theory for processing signals in the space $\ell^2(G)$ where $G$ denotes a countable discrete abelian group. Working in this general setting allows us to study all the classical groups associated with filter banks in digital signal processing in one go. Thus, unidimensional (setting $\ell^2(G)=\ell^2(\mathbb{Z})$), multidimensional ($\ell^2(G)=\ell^2(\mathbb{Z}^d)$), cyclic filter banks ($\ell^2(G)=\ell^2(\mathbb{Z}_s)$), as well as  filter banks processing signals in the spaces $\ell^2(\mathbb{Z}^d\times \mathbb{Z}_s)$, $\ell^2(\mathbb{Z}_r \times \mathbb{Z}_s)$, $\ell^2(\mathbb{Z}_s^d)$ or $\ell^2(\mathbb{Z}_r \times \mathbb{Z}_s \times \mathbb{Z}_v)$ are englobed in the present study.

The proposed abstract group approach is not just a unified way of dealing with classical discrete groups $\mathbb{Z}, \mathbb{Z}^d$ or $\mathbb{Z}_s$; it also allows us to deal with  products of these groups.  This has been pointed out in \cite{ole:15} and it has consequences from a practical point of view: for example, multichannel video signal involves the group $\mathbb{Z}^d \times \mathbb{Z}_s$, where $d$ is the number of channels and $s$ the number of pixels of each image.
Hence the availability of an abstract filter bank theory becomes a useful tool to englobe different digital signal processing problems. 

Besides, nowadays there exists a mathematical literature  dealing with abstract or applied mathematical problems which are studied from a theoretical  groups point of view. See, in particular, Refs.~\cite{barbieri:15,cabrelli:10,cabrelli:12,dodson:07,ole:15,garcia:16} where shift-invariant spaces, Fourier-like frames or sampling problems are considered on LCA groups.
An introduction to group theory and  symmetries in signal processing can be found in Ref.~\cite{sinha:10}.

Classical filter banks have turned out to be  very useful in digital signal processing and in wavelet theory (see, for instance, \cite{mallat:09,strang:96,vaid:93,vetterli:95} and references therein).  One of the main reasons why filter banks has become  so useful has been the use of the polyphase analysis, first carried out by Vetterli \cite{vetterli:86} and Vaidyanathan \cite{vaid:87}, which simplifies considerably the theory, and is especially convenient in their practical design. 
The original  filter bank theory for unidimensional signals in $\ell^2(\mathbb{Z})$ was extended for multidimensional filter banks (see, for instance, \cite{vetterli:92,vaid:93,viscito:91}), as well as for cyclic filter banks \cite{vaid:97,vaid:99}.

Also, associated to a unidimensional analysis filter bank there is a sequence of shifts $\big\{T_nf_k:=f_k(\cdot-n)\big\}_{k=1, 2,\ldots,K;\,n\in \mathbb{Z}}$ of  $K$ elements $f_k$ in $\ell^2(\mathbb{Z})$. The frame property of this sequence give information about the corresponding filter bank: its dual frames provide synthesis filter banks, and its frame bounds  provide  information on the filter bank stability. See, for instance, Refs.~\cite{bol:98,chai:07,johnson:08,vetterli:98,fickus:13} for the unidimensional $\ell^2(\mathbb{Z)}$ setting.

In this paper we introduce the filter bank concept in the setting of a discrete abelian group $G$, and we generalize the polyphase representation for classical filter banks to our setting. This polyphase representation provide a suitable perfect reconstruction condition.  Besides, we extend the frame analysis to this new $\ell^2(G)$ setting. In particular and as far as we know we carry out the first frame analysis for multidimensional  filter banks.

Although our study is done in the polyphase domain, for the sake of completeness, we also include the filter bank representation in the modulation domain, as well as the relationship between polyphase and modulation matrices. The modulation matrix in the group setting was firstly introduced in \cite{Behmard:99}.

The paper is organized as follows: Section \ref{section2}  introduces the properties of  Fourier transform for discrete abelian group used along the article. Section \ref{section3} contains the main results in the paper: we provide a polyphase representation of a filter bank in $\ell^2(G)$ obtaining a perfect reconstruction condition; we derive the corresponding frame analysis obtaining the optimal frame bounds; we also include the filter bank representation in the modulation domain and the relationship between polyphase and modulation matrices. Finally, in Section \ref{section4} we apply the results in Section \ref{section3} to a wide variety of examples.
\section{Some preliminaries on harmonic analysis on groups}
\label{section2}
The results about harmonic analysis on  locally compact abelian (LCA) groups are borrowed from Ref.~\cite{folland:95}; see also \cite{Hewitt:70} or \cite{rudin:90}. Note that, in particular, a countable discrete abelian group is a second countable Hausdorff LCA group. 
\subsection{Convolutions}
Let $G$ be a countable discrete abelian  group with the  operation group denoted by $+$. For, $1\le p<\infty$, $\ell^p(G)$ denotes the set of functions
$x:G\mapsto \mathbb{C}$ such that
$\|x\|_p^p:=\sum_{n\in G} |x(n)|^p < \infty$. For $x,y \in \ell^2(G)$ we define its {\em convolution} as
\[
(x\ast y)(m):= \sum_{n\in G} x(n)
y(m-n)\,, \quad m\in G\,.
\]
The series above converges absolutely for any $m\in G$ \, \cite[Proposition 2.40]{folland:95}. According to \cite[Proposition 2.39]{folland:95}, if $x\in \ell^2(G)$ and $y\in \ell^1(G)$ then $x\ast y\in \ell^2(G)$ and  \begin{equation}
\label{acotacion}
\|x\ast y\|_2\le \|x\|_2 \, \|y\|_1.
\end{equation}
\subsection{The Fourier transform}
Let $\mathbb{T}=\{z\in \mathbb{C}: |z|=1\}$ be the unidimensional torus. 
We said that $\xi:G\mapsto \mathbb{T}$ is a character of $G$ if $\xi(n+m)=\xi(n)\xi(m)$  for all $n,m\in G$. We denote $\xi(n)=\langle n,\xi \rangle$.  Defining $(\xi+\gamma)(n)=\xi(n)\gamma(n)$,  the set of characters $\widehat{G}$ with the operation $+$ is a group, called the dual group of $G$. For $x\in \ell^1(G)$ we define its {\em Fourier transform} as
\[
X(\xi)=\widehat{x}(\xi):=\sum_{n\in G}x(n) \overline{\langle n,\xi \rangle},\quad \xi\in \widehat{G}.
\]

It is known \cite[Theorem 4.5]{folland:95} that $\widehat{\mathbb{Z}}\cong \mathbb{T}$, with $\langle n,z \rangle = z^n$, and
$\widehat{\mathbb{Z}}_s\cong \mathbb{Z}_s:=\mathbb{Z}/s\mathbb{Z}$, with $\langle n,m \rangle = W_s^{nm}$, where $W_s=e^{2\pi i/s}$.
Thus, the Fourier transform on  $\mathbb{Z}$  is
the $z$-transform, \[X(z)=\sum_{n\in \mathbb{Z} } x(n) z^{-n}\] and the Fourier transform on $\mathbb{Z}_s$ is
the s-point DFT, \[X(m)=\sum_{n\in \mathbb{Z}_s } x(n) W_s^{-nm}.\]

There exists a unique  measure, called the Haar measure, $\mu$ on  $\widehat{G}$ satisfying  $\mu(\xi+ E)=\mu(E)$,  for every Borel set $E\subset \widehat{G}$ \cite[Section 2.2]{folland:95},
and $\mu(\widehat{G})=1$. We denote 
$\int_{\widehat{G}} X(\xi) d\xi=\int_{\widehat{G}} X(\xi) d\mu(\xi)$.
If $G=\mathbb{Z}$, 
\[
\int_{\widehat{G}} X(\xi) d\xi=\int_{\mathbb{T}} X(z) dz= \frac{1}{2\pi}\int_0^{2\pi} X(e^{iw}) dw\,,
\]
and if $G=\mathbb{Z}_s$, \[\int_{\widehat{G}} X(\xi) d\xi=\int_{\mathbb{Z}_s} X(n) dn= \frac{1}{s}\sum_{n\in \mathbb{Z}_s} X(n)\,.\]

For $1\le p<\infty$,  $L^p(\widehat{G})$ denotes the set of measurable functions
$X:\widehat{G}\mapsto \mathbb{C}$ such that
$
\|X\|_p^p:=\int_{\widehat{G}} |X(\xi)|^p d\xi < \infty
$. The Fourier transform on $\ell^1(G)\cap \ell^2(G)$ is an isometry on a dense subspace of $L^2(\widehat{G})$. Thus, by Plancherel Theorem it can be extended in a unique manner to a unitary operator  of $\ell^2(G)$ onto $L^2(\widehat{G})$ \cite[p. 99]{folland:95}.

If $x\in \ell^1(G)$ and $X\in L^1(\widehat{G})$ then 
 \[
x(n)=\int_{\widehat{G}}  X(\xi) \langle n,\xi\rangle d\xi,\quad n\in G \qquad \text{(Inversion Theorem\,\,\cite[Theorem 4.32]{folland:95})}
\]
and,  if $x\in \ell^2(G)$ and $h\in \ell^1(G)$ then 
\[
(x\ast h)^{\wedge}(\xi)= X(\xi)H(\xi),\quad \text{a.e.}\,\, \xi\in \widehat{G} \qquad \text{\cite[Theorem 31.27]{Hewitt:70}}.
\]
If $G_1,\ldots G_d$ are abelian discrete groups then the dual group of the product group is 
\[
\big(G_1 \times\ldots \times G_d\big)^{\wedge}\cong \widehat{G}_1 \times\ldots \times \widehat{G}_n \qquad \text{\cite[Proposition 4.6]{folland:95}},
\]  
with $\big\langle\, (x_1,x_2,\ldots,x_d)\, ,\, (\xi_1,\xi_2\ldots,\xi_d)\, \big\rangle = \langle x_1,\xi_1\rangle \langle x_2,\xi_2\rangle\cdots \langle x_d,\xi_d\rangle$.
Hence, $\widehat{\mathbb{Z}^d}\cong \mathbb{T}^d$ and the corresponding Fourier transform is
\[
X(z)= \sum_{n\in \mathbb{Z}^d} x(n)z^{-n},\quad z=(z_1,\ldots,z_d)\in \mathbb{T}^d,
\]
where $z^n= z_1^{n_1}\ldots z_d^{n_d}$. Besides,
$\widehat{\mathbb{Z}_s \times \mathbb{Z}_r }\cong \mathbb{Z}_s \times \mathbb{Z}_r$ and the corresponding Fourier transform is
\[
X(m)= \sum_{n\in \mathbb{Z}_s\times \mathbb{Z}_r} x(n) W_s^{-n_1 m_1}W_r^{-n_2 m_2},\quad m=(m_1,m_2)\in \mathbb{Z}_s\times \mathbb{Z}_r.
\]
\section{Filter banks on discrete groups}
\label{section3}

Let us begin the section giving a short introduction to the {\em polyphase transform} which is the appropriate tool for  analyzing and designing a classical filter bank.  In the $\ell^2(\mathbb{Z})$ setting, a filter bank involves an $L$-fold decimator, also called downsampler or compressor, which takes the input signal 
$\{x(n)\}_{n\in \mathbb{Z}}$ and  produces the output $(\downarrow_L x)=\{x(Ln)\}_{n\in \mathbb{Z}}$, where the sampling period $L$ is a natural number.  The {\em polyphase transform} of  the signal $\{x(n)\}_{n\in \mathbb{Z}}$ is the $L$-dimensional vector which entries are the $z$-transform of the so called {\em polyphase components} $\big\{x(Ln+l)\big\}_{n\in \mathbb{Z}}$, $l=0,1,\ldots,L-1$, of the signal $x$. Namely, 
\[
\mathbf{X}(z)=\Big[\sum_{n\in \mathbb{Z}} x(Ln+l)z^{-n}\Big]_{l=0,1,\ldots,L-1}.
\]
A filter bank designed to process signals in $\ell^2(G)$, where $G$ is a countable discrete abelian group, should involve an $M$-decimator taking the input signal $\{x(n)\}_{n\in G}$ and  producing as output the restriction of $x$ to a subgroup $M$ of $G$ of finite index $L$ (also called a {\em lattice} of $G$), i.e., 
$(\downarrow_M x):=\{x(n)\}_{n\in M}$.   In order to generalize the  polyphase representation to this setting (see \cite{bol:97}), it comes naturally to define the polyphase transform of  $\{x(n)\}_{n\in G}$ as the $L$-dimensional vector whose entries are the $M$-Fourier transform (the Fourier transform  with respect to the subgroup $M$)  of the polyphase components $\{x(m+\ell)\}_{m\in M}$, $\ell \in \mathcal{L}$, where $\mathcal{L}$ is a set of representatives of the cosets of $M$. Namely (see the details below),
\begin{equation}
\label{poli}
\mathbf{X}(\gamma)=\Big[\sum_{m\in M} x(m+\ell) \overline{\langle m, \gamma\rangle} \Big]_{\ell\in \mathcal{L}},\qquad \gamma\in \widehat{M}.
\end{equation}
When $G=\mathbb{Z}^d$ the above transform becomes  that used in multidimensional filter banks \cite{vetterli:92,vaid:93,viscito:91}, while $G=\mathbb{Z}_{s}$ yields the transform used in cyclic filter banks \cite{vaid:97,vaid:99}. It is also worth to note that the considered polyphase transform \eqref{poli} can be obtained, as a particular case of the Zak transform for LCA groups (see, for instance, \cite{barbieri:15b, hernandez:10,kaniuth:98}), which generalizes the classical Zak transform; specifically, $\mathbf{X}(\gamma)= \big[Zx(\ell,\gamma) \big]_{\ell\in \mathcal{L}}$.
\subsection{The lattice $M$}

Throughout the article, we assume that {\em $M$ is a subgroup of $G$ with finite index $L$}; we fix a set $\mathcal{L}=\{\ell_0,\ldots,\ell_{L-1}\}$ of {\em representatives of the cosets of $M$}, i.e., the group $G$ can be decomposed as
\[
G= (\ell_0+M) \cup (\ell_1+M) \cup \ldots \cup (\ell_{L-1}+M)
\]
with $(\ell_r+M) \cap (\ell_{r'}+M)=\varnothing $ for $r\neq r'$ (the set $\mathcal{L}$ is also called a transversal or a  section of $M$). 
For instance, for $G=\mathbb{Z}$ and $M=L\mathbb{Z}$ we can take
$\mathcal{L}=\{0,1,\ldots,L-1\}$ since
\[
\mathbb{Z}= L\mathbb{Z} \cup (1+L\mathbb{Z}) \cup \cdots \cup (L-1+L\mathbb{Z}).
\] 
We denote
by $\ast_{_M}$ the convolution with respect to the subgroup $M$, i.e.,
\[
 \big(c \ast_{_M} d\big) (n):= \sum_{m\in M} c(m)\,d(n-m)\,, \quad n\in M\,.
\]
\subsection{The $M$-Fourier transform}
The {\em annihilator of $M$} is the subgroup of $\widehat{G}$ given by
$M^\perp:=\big\{\xi \in \widehat{G} : \langle m,\xi \rangle=1\,\,\text{for all}\,\, m\in M \big\}$, which  has $L$ elements \cite[Section 4.3]{folland:95}. 
We have that 
\[\widehat{M}\cong \widehat{G}/M^\perp \quad \text{with}\quad
 \big\langle m,\xi+ M^\perp \big\rangle =  \langle m,\xi \rangle\qquad\text{\cite[Theorem 4.39]{folland:95} }.
\]
We denote by $C(\xi + M^\perp)$ or $\widehat{c}(\xi + M^\perp)$ the Fourier transform of a function $c$ in the group $M$, 
 i.e.,
\[
 C(\xi+ M ^\perp)= \sum_{m\in M} c(m) \overline{\langle m, \xi+ M^\perp\rangle}= \sum_{m\in M} c(m) \overline{\langle m, \xi\rangle}.
\]
As we said above, our polyphase representation relies on this transform. Thus,
in many occasions to simplify the notation we denote the 
characters of $\widehat{M}$ by $\gamma$ instead of $\xi + M^\perp$.
To prevent confusions, we call $C(\gamma)$ the {\em $M$-Fourier transform of $c$}.
\subsection{The filter bank }
For a complex function $x$ with domain in $G$, we denote its  {\it restriction to the subgroup $M$} as
 \[
 (\downarrow_M x)(m)=x(m),\quad m\in M.
 \] 
 For a complex function $x$ with domain $M$ we define the {\it expander to $G$} as
\[
(\uparrow_M x)(n)= \begin{cases}x(n),& n\in M\\
0,& n\notin M.\end{cases}
\] 
Throughout this paper we consider the {\em $K$-channel filter bank} represented in Fig. \ref{f1}, i.e.,
 \[
 c_k=\big( \downarrow_M (x\ast h_k)\big)\quad \text{and}\quad 
 y= \sum_{k=1}^K (\uparrow_M c_k) \ast g_k\,,
 \]
 where $h_k$, $k=1, 2, \dots, K$, are the {\em analysis filters}, and  $g_k$, $k=1, 2, \dots, K$, are the {\em synthesis filters}. 
 Equivalently, we have the input-output expression:
 \begin{equation}
\label{y}
y(n)=\sum_{k=1}^K \sum_{m\in M} (x \ast h_k) (m) \, g_k(n-m),\quad  n\in G\,.
\end{equation}
In the sequel, we assume that  the filters $h_k,g_k\in \ell^1(G)$,\, for $k\in\mathcal{K}$, where for notational ease we denote $\mathcal{K}:=\{1, 2, \ldots,K\}$. This assumption guarantees the convergence of series involved in \eqref{y} for any $x\in \ell^2(G)$. Indeed, from \eqref{acotacion} we have that $x\ast h_k \in \ell^2(G)$ and then, again from \eqref{acotacion}, we have that the series in \eqref{y} converges absolutely and $y \in \ell^2(G)$.

 \medskip
  
 \begin{figure}
\begin{center}
\includegraphics[bb=5cm 23.4cm 12cm 25.6cm]
{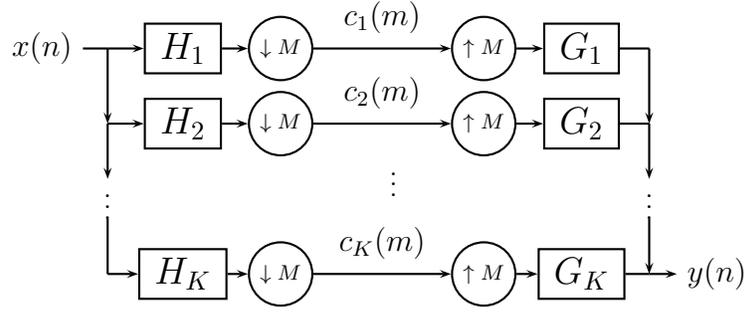}
\vspace{1cm}

 \caption{Scheme for a K-channel filter bank}\label{f1}
\end{center}
\end{figure}

\subsection{Polyphase analysis}
For $k=1, 2, \dots, K$ and $\ell \in \mathcal{L}$ we define the {\em polyphase components} of $x$, $h_k$, $g_k$ and $y$ as
\[
\begin{split}
& x_{\ell}(m):=x(m+\ell),\quad y_{\ell}(m):=y(m+\ell),\\& h_{k,\ell}(m):=h_k(m-\ell), \quad g_{\ell,k}(m):=g_k(m+\ell), \quad m\in M,
\end{split}
\] 
and we denote their $M$-Fourier transforms
by
$
X_{\ell}(\gamma), Y_{\ell}(\gamma), H_{k,\ell}(\gamma), G_{\ell,k}(\gamma)$ respectively. 

For any $x\in \ell^2(G)$ and $m\in M$, we have
\[
\begin{split}
c_k(m)&=(x \ast h_k) (m) = \sum_{n\in G} x(n)\, h_k(m-n)=
\sum_{\ell\in \mathcal{L}}\sum_{n\in M} x(n+\ell)\, h_k(m-n-\ell)\\&=
\sum_{\ell\in \mathcal{L}}\sum_{n\in M} x_{\ell}(n)\, h_{k,\ell}(m-n)= \sum_{\ell\in \mathcal{L}} (x_{\ell} \ast_{_M}  h_{k,\ell}) (m).
\end{split}
\]
All the series above converge absolutely since we have assumed that $h_k\in \ell^1(G)$. Moreover, $\mathrm{c}_k \in \ell^2(M)$ since $x \ast h_k\in \ell^2(G)$. Taking the $M$-Fourier transform, we obtain
\begin{equation}
\label{C}
C_k(\gamma)=  \sum_{\ell\in \mathcal{L}}  H_{k,\ell}(\gamma)X_\ell(\gamma).
\end{equation}
Thus, we have the matrix expression
\begin{equation}
\label{analisis}
\mathbf{C}(\gamma)= \mathbf{H}(\gamma) \mathbf{X}(\gamma)\quad \text{a.e.}\,\, \gamma \in \widehat{M},
\end{equation}
where 
\begin{equation}
\label{matriz polifase analisis}
\mathbf{C}(\gamma)=[C_k(\gamma)]_{k\in \mathcal{K}}\, ,\quad
\mathbf{X}(\gamma)=[X_\ell(\gamma)]_{\ell\in \mathcal{L}}\, ,\quad
\mathbf{H}(\gamma)=\big[H_{k,\ell}(\gamma)\big]_{k\in \mathcal{K},\,\ell\in \mathcal{L}}\, .
\end{equation}
\noindent Above, $\mathbf{C}(\gamma)$ and $\mathbf{X}(\gamma)$ denote  column vectors, i.e., $\mathbf{C}(\gamma)=[
C_{1}(\gamma),\ldots,C_{K}(\gamma)
]^\top$ and $\mathbf{X}(\gamma)=[
X_{\ell_0}(\gamma),\ldots ,X_{\ell_{L-1}(\gamma)}
]^\top$, and $\mathbf{H}(\gamma)$ is a $K\times L$ matrix.

The polyphase components of the output $y$ can be written as
\[
\begin{split}
y_{\ell}(m)&=y(m+\ell)=\sum_{k=1}^K \sum_{n\in M} c_k(n) g_k(m+\ell-n)
\\&=\sum_{k=1}^K \sum_{n\in M} c_k(n)\, g_{\ell,k}(m-n)= \sum_{k=1}^K (c_k \ast_{_M} g_{\ell,k})(m).
\end{split}
\]
Taking the $M$-Fourier transform, we obtain 
\[
Y_\ell(\gamma)= \sum_{k=1}^K G_{\ell,k}(\gamma) C_k(\gamma)\qquad \text{a.e.}\,\, \gamma \in \widehat{M},
\]
which can be written as
\begin{equation}
\label{sintesis}
\mathbf{Y}(\gamma)= \mathbf{G}(\gamma) \mathbf{C}(\gamma)\quad \text{a.e.}\,\, \gamma \in \widehat{M},
\end{equation}
where
\begin{equation}
\label{matriz polifase sintesis}
\mathbf{Y}(\gamma)=[Y_\ell(\gamma)]_{\ell\in \mathcal{L}}\, ,\quad
\mathbf{G}(\gamma)=\big[G_{\ell,k}(\gamma)\big]_{\ell\in \mathcal{L}, \, k\in \mathcal{K}}\,.
\end{equation}
Thus, from \eqref{analisis} and \eqref{sintesis}, we have
\begin{equation}
\label{analisis sintesis polifase}
\mathbf{Y}(\gamma)\, =  \,\mathbf{G}(\gamma)\, \mathbf{H}(\gamma) \,\mathbf{X}(\gamma)\quad \text{a.e.}\,\, \gamma \in \widehat{M}.
\end{equation}
On the other hand, we consider in the following proposition a generalization to discrete groups
of the {\em polyphase transform}:

\medskip
\begin{pro}
\label{polifase transform} 
The polyphase transform $\mathcal{P}:\ell^2(G) \rightarrow L^2(\widehat{M})\times\cdots \times  L^2(\widehat{M})\quad (L\,\,  \text{times}) $ defined by $\mathcal{P}(x):=\mathbf{X}= [X_\ell]_{\ell\in \mathcal{L}},$ is a unitary operator.
\end{pro}
\begin{proof}
For $x,y\in \ell^2(G)$ we have
\[
\begin{split}
\langle x,y\rangle_{\ell^2(G)} &= \sum_{\ell\in \mathcal{L}} \sum_{m\in M} x(m+\ell)\overline{y(m+\ell)}=
\sum_{\ell\in \mathcal{L}} \sum_{m\in M} x_{\ell}(m)\overline{y_{\ell}(m)}= \sum_{\ell\in \mathcal{L}} \langle x_{\ell},y_{\ell}\rangle_{\ell^2(M)}\\
&=\sum_{\ell\in \mathcal{L}} \langle X_\ell,Y_\ell\rangle_{L^2(\widehat{M})}=  \langle \mathbf{X},\mathbf{Y}\rangle_{L^2(\widehat{M})\times\cdots \times  L^2(\widehat{M})}.
\end{split}
\]
Then $\mathcal{P}$ is a isometry.
Besides, for any $\mathbf{X}\in L^2(\widehat{M})\times\cdots \times  L^2(\widehat{M})$, since the $M$-Fourier transform is a surjective isometry between  $\ell^2(M)$ and $L^2(\widehat{M})$, there exists a function  $x$ such that its polyphase  components  $[X_{\ell}]_{\ell\in \mathcal{L}}$ coincides with $\mathbf{X}$. Hence,  
$\mathcal{P}$ is surjective.
\end{proof}

\medskip

By using Proposition \ref{polifase transform}, from \eqref{analisis sintesis polifase}
we easily deduce:
\medskip

\begin{theorem} 
\label{PR} 
The filter bank defined by \eqref{y} satisfies the perfect reconstruction property, i.e, $y=x$ for all $x\in \ell^2(G)$ if and only if
$\mathbf{G}(\gamma) \,\mathbf{H}(\gamma) =\mathbf{I}_L$\, for all  $\gamma\in \widehat{M}$, where $\mathbf{I}_L$ denotes the identity matrix of order $L$.
\end{theorem}
\begin{proof}
Having in mind Proposition \ref{polifase transform}  and \eqref{analisis sintesis polifase} the filter bank satisfies the perfect reconstruction property if and only if $\mathbf{G}(\gamma)\mathbf{H}(\gamma)=\mathbf{I}_L$ a.e. $\gamma \in \widehat{M}$. Since we have assume that $h_k$ and $g_k$ belong to $\ell^1(G)$, their polyphase components, 
$h_{k,\ell}$ and $g_{\ell,k}$ belong to $L^1(M)$. Then their $M$-Fourier transform are continuous \cite[Proposition 4.13]{folland:95}. Hence, the entries of $\mathbf{G}(\gamma)\mathbf{H}(\gamma)$ are continuous. Therefore, 
$\mathbf{G}(\gamma)\mathbf{H}(\gamma)=\mathbf{I}_L$ a.e. $\gamma \in \widehat{M}$ if and only if 
$\mathbf{G}(\gamma)\mathbf{H}(\gamma) =\mathbf{I}_L$  for all $\gamma\in \widehat{M}$.
\end{proof}

\medskip

It is easy to check that between the polyphase transform and the Fourier transform, there exists the relationship
\[
X(\xi)= \mathbf{p}^\top(\xi) \mathbf{X}(\xi+ M^\perp),\,\, \xi \in \widehat{G}, \quad \text{where}\quad \mathbf{p}(\xi)= \big[ \, \overline{\langle \ell,\xi \rangle}\, \big]_{\ell\in \mathcal{L}}\,.
\]
Then, from \eqref{analisis sintesis polifase}  the Fourier transform of the output $y$ is expressed as
\[
Y(\xi)= \mathbf{p}^\top(\xi) \,\mathbf{G}(\xi+ M^\perp) \mathbf{H}(\xi+ M^\perp) \,\,\mathbf{X}(\xi + M^\perp)\quad \text{a.e.}\,\, \xi\in  \widehat{G}.
\]
\subsection{Frame analysis}

For $m\in M$, we denote the {\em translation operator by $m$} as  $(T_m f)(n):=f(n-m)$, $n\in G$, and the {\em involution of $f$} as $\widetilde{f}(n):= \overline{f(-n)}$, $n\in G$. Then, for $k=1, 2, \dots, K$
\begin{equation}
\label{escalar convolucion}
c_k(m):=(x \ast h_k) (m)= \sum_{n\in G} x(n) \,h_k(m-n)=  \langle x, T_m \widetilde{h}_k \rangle_{\ell^2(G)},\quad m\in G\,,
\end{equation}
and if, for notational ease, we denote $f_k:=\widetilde{h}_k$, $k=1, 2, \dots, K$, the expansion \eqref{y} representing the filter bank can be written as
\[
y= \sum_{k=1}^K \sum_{m\in M} \langle x, T_m f_k \rangle_{\ell^2(G)}\,\, T_m g_k\,.
\]
Thus, the filter bank in Fig.~\ref{f1} is related to the sequences  $\big\{T_m f_k\big\}_{k\in\mathcal{K},\, m\in M}$ and
$\big\{T_m g_k\big\}_{k\in\mathcal{K},\, m\in M}$.  The following results provide the frame properties of these sequences. In Ref.~\cite{ole:03} the reader can find the main properties of frames and Riesz bases. Recall that we have assumed that  $h_k\in \ell^1(G)$ which is equivalent to assume that $f_k\in \ell^1(G)$. 

\begin{theorem}
\label{dual frames}  
The sequences $\big\{T_m f_k\big\}_{k\in\mathcal{K}, \,m\in M}$ and
$\big\{T_m g_k\big\}_{k\in\mathcal{K}, \,m\in M}$ are dual frames for $\ell^2(G)$ if and only if $\,\mathbf{G}(\gamma)\, \mathbf{H}(\gamma)=\mathbf{I}_L$ for all $\gamma \in \widehat{M}$.
\end{theorem}
\begin{proof}
By using \eqref{escalar convolucion} and \eqref{acotacion}, for each $k=1, 2, \dots, K$ we obtain that
\[
\begin{split}
&\sum_{m\in M}|\langle x,T_m f_k \rangle|^2\le \sum_{n\in G} |\langle x,T_n f_k \rangle|^2=\sum_{n\in G} |x \ast h_k (n)|^2\\&=  \|x \ast h_k\|^2_{2}\le   \|x\|_{2}^2  \|h_k\|^2_{1}\,, \,\, \text{for all $x\in \ell^2(G)$}\,.
\end{split}
\]
Hence, $\big\{T_m f_k\big\}_{k\in \mathcal{K}, m\in M}$ is a Bessel sequence for $\ell^2(G)$. Analogously one proves that the sequence
$\big\{T_m g_k\big\}_{k\in \mathcal{K}, m\in M}$ is a Bessel sequence for $\ell^2(G)$. Having in mind Lemma 5.6.2 in \cite{ole:03}, the result is now a consequence of Theorem \ref{PR}.
\end{proof}

\medskip

Let $\mathbf{H}^*(\gamma)$  denote the transpose conjugate of the matrix $\mathbf{H}(\gamma)$. 

\medskip

\begin{theorem}\label{cotas} The sequence
$\big\{T_m f_k\big\}_{k\in\mathcal{K},\,m\in M}$ is a frame for $\ell^2(G)$ if and only if Rank $\mathbf{H}(\gamma)=L$ for all $\gamma \in \widehat{M}$.
In this case, the optimal frame bounds are
\[A= \min_{\gamma\in \widehat{M}}\big[ \lambda_{\min}(\gamma)\big]\quadÊ\text{and}\quad B= \max_{\gamma\in \widehat{M}}\big[\lambda_{\max}(\gamma)\big]\,,
\]
where $\lambda_{\min}(\gamma)$  and $\lambda_{\max}(\gamma)$ are  the smallest and the largest eigenvalue of the matrix $\mathbf{H}^*(\gamma)\mathbf{H}(\gamma)$. 

In case $\big\{T_m f_k\big\}_{k\in\mathcal{K},\,m\in M}$is a frame for $\ell^2(G)$, its canonical dual frame is 
$\big\{T_m \overline{f}_k\big\}_{k\in\mathcal{K},\,m\in M}$ where $\overline{f}_k=\mathcal{P}^{-1} (\mathbf{H}^*\mathbf{H})^{-1}\mathcal{P} f_k$,\, $k=1, 2, \dots, K$, where $\mathcal{P}$ denotes the polyphase transform in Prop. \ref{polifase transform}.
\end{theorem} 
\begin{proof}
First,  notice that $\lambda_{\min}(\gamma)$  and $\lambda_{\max}(\gamma)$ have a minimum and a maximum value over $\widehat{M}$. Indeed,  since  $h_k\in \ell^1(G)$, the entries of $\mathbf{H}^*(\gamma)\mathbf{H}(\gamma)$ are continuous functions  \cite[Proposition 4.13]{folland:95} and then  $\lambda_{\min}(\gamma)$  and $\lambda_{\max}(\gamma)$ are real continuous functions  (see \cite{zedek:65}). Besides, since $M$ is discrete, $\widehat{M}$ is compact \cite[Proposition 4.4]{folland:95}.

In the proof of Theorem \ref{dual frames} we have showed that $\big\{T_m f_k\big\}_{k\in \mathcal{K}, m\in M}$ is a Bessel sequence. Now, we obtain a representation in the polyphase domain for its frame operator
\[
Sx =\sum_{k=1}^K \sum_{m\in M} \langle x, T_m f_k \rangle \,\, T_m f_k, \quad x\in \ell^2(G)\,.
\]
Indeed, when 
 $g_k(n)=f_k(n)=\overline{h_k(-n)}$, then
 $\mathbf{G}=\mathbf{H}^*$, and the representation \eqref{analisis sintesis polifase} reads
\[
\big[\mathcal{P}Sx\big](\gamma)=\mathbf{H}^*(\gamma)\mathbf{H}(\gamma) \mathbf{X}(\gamma).
\]
By using Proposition \ref{polifase transform}, we get
\[
\begin{split}
&\sum_{k=1}^K \sum_{m\in M} \big|\langle x, T_m f_k \rangle\big|^2= \langle Sx, x \rangle_{\ell^2(G)}
=\big\langle \mathcal{P}Sx, \mathcal{P}x \big\rangle_{L^2(\widehat{M})\times\ldots\times L^2(\widehat{M})}\\ &=
\int_{\widehat{M}} \mathbf{X}^*(\gamma)\big[\mathcal{P}Sx\big](\gamma)d\gamma = \int_{\widehat{M}} \mathbf{X}^*(\gamma) \mathbf{H}^*(\gamma)\mathbf{H}(\gamma)\mathbf{X}(\gamma)\, d\gamma.
\end{split}
\]
Hence,
\[
\begin{split}
&\sum_{k=1}^K \sum_{m\in M} \big|\langle x, T_m f_k \rangle\big|^2 \ge
\int_{\widehat{M}} \lambda_{\min}(\gamma) |\mathbf{X}(\gamma)|^2 d\gamma \ge
A \int_{\widehat{M}}  |\mathbf{X}(\gamma)|^2 d\gamma \\&=A \|\mathbf{X}\|^2_{L^2(\widehat{M})\times\ldots\times L^2(\widehat{M})} = A \|x\|^2_2.
\end{split}
\]
Let $J>A$;  there exists a subset  $\Omega \subset \widehat{M}$ with positive measure such that $\lambda_{\min}(\gamma)<J$ for $\gamma\in \Omega$. Let $\mathbf{X}(\gamma)$ be 
 equal to $0$ when $\gamma\notin \Omega$ and equal to a unitary eigenvector of $\mathbf{H}^*(\gamma)\mathbf{H}(\gamma)$ corresponding to  $\lambda_{\min}(\gamma)$ when $\gamma\in\Omega$. Notice that $\mathbf{X}\in L^2(\widehat{M})\times\ldots\times L^2(\widehat{M})$  since $\|\mathbb{F}\|^2_{L^2(\widehat{M})\times\ldots\times L^2(\widehat{M})}=$ measure$(\Omega) \le 1$. The function $x=\mathcal{P}^{-1}\mathbf{X}$ satisfies
\[
\sum_{k=1}^K \sum_{m\in M} \big|\langle x, T_m f_k \rangle\big|^2=
\int_{\Omega} \mathbf{X}^*(\gamma)\mathbf{H}^*(\gamma)\mathbf{H}(\gamma)\mathbf{X}(\gamma)d\gamma= \int_{\Omega} \lambda_{\min}(\gamma) \mathbf{X}^*(\gamma)\mathbf{X}(\gamma)d\gamma\le  J \|x\|^2\,.
\]
Therefore, the sequence $\big\{T_m f_k\big\}_{k\in \mathcal{K},m\in M}$ is a frame for $\ell^2(G)$ if and only if $A>0$, and in this case the lower optimal bound is $A$.
In the same way it can be proved that $B$ is the optimal Bessel bound. Since $\lambda_{\min}(\gamma)$ is a continuous function,  $A>0$ if and only if $\lambda_{\min}(\gamma)>0$ for all $\gamma\in \widehat{M}$ which is equivalent to be the rank of  $\mathbf{H}(\gamma)$ equal to $L$ for all $\gamma \in \widehat{M}$.

It is easy to check that
$ST_m x=T_m S x$. The  canonical dual frame is given by (see \cite[Lemma 5.1.1]{ole:03})
\[
S^{-1}T_mf_k= T_m S^{-1}f_k =T_m \mathcal{P}^{-1}\mathcal{P}S^{-1}f_k=T_m\mathcal{P}^{-1}(\mathbf{H}^*\mathbf{H})^{-1}\mathcal{P}f_k\,.
\]
\end{proof}

\medskip

The synthesis matrix $\mathbf{G}(\gamma)$ corresponding to the  canonical dual frame is $[\mathbf{H}(\gamma)^*\mathbf{H(\gamma)}]^{-1}\mathbf{H}^*(\gamma)$, which coincides with the Moore-Penrose pseudoinverse $\mathbf{H}^\dag(\gamma)$ of the analysis matrix $\mathbf{H}(\gamma)$.

Analogously, the optimal frame bounds of the dual frame $\big\{T_m g_k\big\}_{k\in\mathcal{K},\,m\in M}$ are 
given by $A_g= \min_{\gamma\in \widehat{M}}\big[ \mu_{\min}(\gamma)\big]$ and $B_g= \max_{\gamma\in \widehat{M}} \big[\mu_{\max}(\gamma)\big]$, where  
$\mu_{\min}(\gamma)$  and $\mu_{\max}(\gamma)$ are the smallest and the largest eigenvalues of  the matrix $\mathbf{G}(\gamma)\mathbf{G}^*(\gamma)$. For the canonical dual frame $g_k=\overline{f}_k$, we have that $A_g=1/B$   and $B_g=1/A$\, \cite[Lemma 5.1.1]{ole:03}.

The frame bounds give information about the stability of the filter bank.
Notice that, by its definition, the optimal frames bounds of $\big\{T_m f_k\big\}_{k\in\mathcal{K},m\in M}$ are   the tightest numbers $0<A\le B$ such that
\[
A \|x\|_2^2 \le \sum_{k=1}^K\sum_{m\in M} |c_k(m)|^2= \sum_{k=1}^K\sum_{m\in M} \big|(x\ast h_k)(m)\big|^2\le B \|x\|_2^2\,, \quad x\in \ell^2(G)\,.
\]
Thus $B$ gives a measure of how an error in the input $x$ of the analysis filter bank affects to subband signals $c_k$. For the synthesis, we have that $B_g$ is   the tightest number such that \cite[Theorem 3.2.3]{ole:03}
\[
\|y\|^2=\Big\|  \sum_{k=1}^K\sum_{m\in M} c_k(m) g_k(\cdot-m) \Big\|^2 \le B_g
\sum_{k=1}^K\sum_{m\in M} |c_k(m)|^2 \,.
\]
Thus $B_g$ gives a measure of how an error in the subband signals $c_k$ affects to the recovered signal $y$. The smallest possible value for $B_g$ is  $1/A$, which correspond to take the canonical dual frame. One can find a sensitivity analysis based on frame bounds in Ref. \cite{bol:98}; see also \cite[p.\,118]{ole:03}.

Having in mind that $A=B$ if and only if $\mathbf{H}^*(\gamma)\,\mathbf{H}(\gamma)= A \, \mathbf{I}_L$ for all $\gamma \in \widehat{M}$, we deduce: 

\medskip
\begin{cor}\label{tight} The sequence
$\big\{T_m f_k\big\}_{k\in\mathcal{K},\,m\in M}$ is a tight frame for $\ell^2(G)$ if and only if there exists $A>0$ such that $\mathbf{H}^*(\gamma)\,\mathbf{H}(\gamma)= A \, \mathbf{I}_L$ for all $\gamma \in \widehat{M}$.
In this case, the frame bound is $A$.
\end{cor}

\medskip

For maximally decimated filter banks, i.e., whenever $L=K$, we have the following result:

\medskip

\begin{theorem}
\label{maximally decimated} Assume that $L=K$.  The sequence
$\big\{T_m f_k\big\}_{k\in\mathcal{K},\,m\in M}$ is Riesz basis for $\ell^2(G)$ if and only if  $\det \mathbf{H}(\gamma)\neq 0$ for all $\gamma \in \widehat{M}$.
In this case, the  optimal Riesz bounds are the constants $A$ and $B$ defined in Theorem~\ref{cotas}.
\end{theorem}

\begin{proof}
If the sequence $\big\{T_m f_k\big\}_{k\in \mathcal{K},m\in M}$ is a Riesz basis then it is a frame. Then, by Theorem~\ref{cotas}, Rank $\mathbf{H}(\gamma)=L=K$, and thus $\det \mathbf{H}(\gamma)\neq 0$, for all $\gamma \in \widehat{M}$.

To prove the reciprocal, assume that $\det \mathbf{H}(\gamma)\neq 0$, for all $\gamma \in \widehat{M}$.
Then  Rank $\mathbf{H}(\gamma)=L$ and, from Theorem \ref{cotas}, $\big\{T_m f_k\big\}_{k\in\mathcal{K},\,m\in M}$ is a frame. Thus, to prove that it is a Riesz basis  it only remains to prove that it has a biorthogonal sequence \cite[Theorem 6.1.1]{ole:03}. Notice that since $|\det \mathbf{H}(\gamma)|$ is continuous on the compact $\widehat{M}$, and $|\det \mathbf{H}(\gamma)|>0$ for all $\gamma \in \widehat{M}$, then there exists $J>0$ such that $|\det \mathbf{H}(\gamma)|>J$ for all $\gamma \in \widehat{M}$. Then the rows of $\mathbf{H}^{-1}(\gamma)$ belong to 
$L^2(\widehat{M})\times\ldots\times L^2(\widehat{M})$.
 We denote by $g_1,\ldots,g_k$ the inverse polyphase transform (see Proposition \ref{polifase transform})  of these rows. Thus $\mathbf{G}(\gamma)$ defined by \eqref{matriz polifase sintesis} is $\mathbf{G}(\gamma)=\mathbf{H}(\gamma)^{-1}$. 
From \eqref{C}, we obtain that the $M$-Fourier transform of $c_{k,k'}=\downarrow_M (g_{k'}\ast h_k)$ is 
\[
C_{k,k'}(\gamma)=\sum_{\ell\in \mathcal{L}} H_{k,\ell}(\gamma)G_{\ell,k'}(\gamma)\,.
\]
Since $\mathbf{G}(\gamma)=\mathbf{H}^{-1}(\gamma)$ we obtain that
$
C_{k,k'}(\gamma)=\delta_{k,k'}
$
Then, having in mind that the inverse $M$-Fourier transform of $C_{k,k}=1$ is the $\delta$ sequence, by using \eqref{escalar convolucion} we obtain
\[
\langle T_{m'}g_{k'},T_m f_k \rangle=
\langle g_{k'},T_{m-m'} f_k \rangle=
\big(g_{k'}\ast h_k\big)(m-m')=c_{k,k'}(m-m')= \delta_{k,k'}\delta_{m,m'}\,,
\]
which proves that the sequence $\big\{T_m f_k\big\}_{k\in\mathcal{K},\,m\in M}$ is a Riesz basis for $\ell^2(G)$. The optimal Riesz bounds are the optimal frame bounds  \cite[Theorem 5.4.1]{ole:03}, and then, from Theorem \ref{cotas}, they are $A$ and $B$.
\end{proof}

\subsection{Modulation Analysis}

Recall that $M^\perp$, the annihilator of $M$,  is a subgroup of $\widehat{G}$ with $L$ elements.

\medskip

\begin{pro} \label{decimacion} For any $x\in \ell^2(G)$, the $M$-Fourier transform of $(\downarrow_M x)$ is
\[(\downarrow_M x)^{\wedge}(\xi + M^\perp)=\dfrac{1}{L} \sum_{\eta\in M^\perp} X(\xi+ \eta),\quad \text{a.e.} \,\, \xi\in \widehat{G}.\] 
\end{pro}

\begin{proof}
If $n\notin M$ we have that there exist $\eta_r\in M^\perp$ such that $\langle n ,\eta_r \rangle \neq 1$  \cite[Proposition 4.38]{folland:95}. Since $M^\perp$ is a group,
\[\sum_{\eta \in M^\perp}  \langle n,\eta \rangle= 
\sum_{\eta \in M^\perp}  \langle n,\eta+\eta_r \rangle=
 \langle n,\eta_r \rangle \sum_{\eta \in M^\perp}  \langle n,\eta \rangle.\]
Therefore
\begin{equation}\label{ortog}
\sum_{\eta \in M^\perp} \langle n,\eta \rangle =\begin{cases}L&n\in M\\0& n\notin M. \end{cases}
\end{equation}
By using this relationship, we obtain
\[
\begin{split}
\qquad \qquad \quad (\downarrow_M x)^{\wedge}(\xi+M^\perp)&= \sum_{m \in M} x(m) \overline{\langle m,\xi \rangle}=\dfrac{1}{L}
\sum_{n \in G} \sum_{\eta \in M^\perp} \overline{\langle n,\eta \rangle} x(n) \overline{\langle n,\xi \rangle}
\\&=\dfrac{1}{L}
\sum_{\eta \in M^\perp} \sum_{n \in G}  x(n) \overline{\langle n,\xi+\eta \rangle}=\dfrac{1}{L} \sum_{\eta \in M^\perp} \widehat{x}(\xi+\eta).
\end{split}
\]
\end{proof}

As a consequence of the above proposition, the $M$-Fourier transform of $c_k=\downarrow_M(x\ast h_k)$ is 
$
C_k(\xi + M^\perp)= \dfrac{1}{L} \sum_{\eta\in M^\perp} X(\xi+\eta) H_k(\xi+\eta)
$.
Hence, denoting $\mathbf{C}=\big[C_k \big]_{k\in \mathcal{K}}$, we have
\begin{equation}
\label{AC analisis}
 \mathbf{C}(\xi + M^\perp)^\top= \dfrac{1}{L}
\mathbf{H}_{\text{mod}}(\xi)\,\mathbf{x}_{\text{mod}}(\xi),
\end{equation}
where 
\begin{equation}
\label{AC matrix}
\mathbf{x}_{\text{mod}}(\xi):=\big[X(\xi+\eta)\big]_{\eta\in M^\perp}\quad\text{and}\quad \mathbf{H}_{\text{mod}}(\xi):=\big[ H_k(\xi+\eta)\big]_{k\in \mathcal{K},\, \eta\in M^\perp}.
\end{equation}
For any $c\in \ell^2(M)$, the Fourier transform of $(\uparrow_M c)$ is $M^\perp$-periodic; specifically, for any $\eta \in M^\perp$,
\[
(\uparrow_M c)^\wedge(\xi+ \eta)=\sum_{n\in G} [\uparrow_M c](n) \overline{\langle n,\xi+ \eta \rangle}
=\sum_{m\in M} c(m) \overline{\langle m,\xi + M^\perp \rangle}=C(\xi+ M^\perp).
\] 
Then the Fourier transform of
\[
y(n)=\sum_{k=1}^K\sum_{m\in M} c_k(m) g_k(n-m)= \sum_{k=1}^K \sum_{l\in G} (\uparrow_M c_k)(l) g_k(n-l)=
\sum_{k=1}^K \big((\uparrow_M c_k)\ast g_k\big)(n)
\]
is $Y(\xi)= \sum_{k=1}^K C_k(\xi+ M^\perp)^\top G_k(\xi)$. From \eqref{AC analisis}, the Fourier transform of the output $y$ to the filter bank in Fig. \ref{f1} is
\[
Y(\xi)= 
\dfrac{1}{L}  \big[G_1(\xi), G_2(\xi), \cdots, G_K(\xi)\big]\, \mathbf{H}_{\text{mod}}(\xi)\,\mathbf{x}_{\text{mod}}(\xi),\quad \xi \in \widehat{G}.
\]
This modulation representation of the output to the filter bank was obtained in \cite{Behmard:99}.

\medskip
\begin{pro}\label{ACpolifase}
The $K\times L$ matrices $\mathbf{H}_{\text{mod}}(\xi)$ and $\mathbf{H}(\xi)$, defined in \eqref{AC matrix} and \eqref{matriz polifase analisis} respectively, are related by
\[
\mathbf{H}_{\text{mod}}(\xi)= \,\mathbf{H}(\xi + M^\perp)  \, \mathbf{D}(\xi)\, \mathbf{W}\quad \text{for all}\, \,\, \xi\in \widehat{G},
\]
where
$
\mathbf{W}=\big[ \langle \ell_i, \eta \rangle \big]_{i=0,1,\ldots,L-1, \,\eta\in M^\perp}\,\,$ and $\,\,
\mathbf{D}(\xi)= \text{diag}\,\big(\langle \ell_0, \xi \rangle, \langle \ell_1, \xi \rangle ,\ldots,\langle \ell_{L-1}, \xi \rangle  \big)
$.
\end{pro}
\begin{proof}
We have
\[
\begin{split}
H_{k}(\xi)&=\sum_{n\in G} h_k(n)  \overline{\langle n, \xi \rangle} =
\sum_{\ell\in \mathcal{L}}\sum_{m\in M} h_k(m-\ell)  \overline{\langle m-\ell, \xi \rangle}\\&=
\sum_{\ell\in \mathcal{L}} \langle \ell, \xi \rangle \sum_{m\in M} h_k(m-\ell)  \overline{\langle m, \xi \rangle}
=
\sum_{\ell\in \mathcal{L}} \langle \ell, \xi \rangle H_{k,\ell}(\xi+M^\perp).
\end{split}
\]
Therefore, 
$H_{k}(\xi+\eta)=\sum_{\ell\in \mathcal{L}} \langle \ell, \xi \rangle H_{k,\ell}(\xi+M^\perp) \langle \ell, \eta \rangle$ for all $\xi \in \widehat{G}$, $\eta\in M^\perp$.
\end{proof}

\medskip
It is worth to note that $\mathbf{W}\mathbf{W}^*=L\, \mathbf{I}_L$ (see \eqref{ortog}); then
$\mathbf{H}(\xi+M^\top)=(1/L) \mathbf{H}_{\text{mod}}(\xi)\mathbf{W}^*\overline{\mathbf{D}(\xi)}$.

\section{Some illustrative examples}
\label{section4}
In this section we consider the filter bank depicted in Fig.~\ref{f1} for different choices of the group $G$ and the lattice $M$. Thus,
we particularize the  general theory in Section \ref{section3} in four different contexts:
\subsection{The case $G=\mathbb{Z}^d$ and $M=\{\mathbf{M} n: n\in \mathbb{Z}^d$\}}

Let $\mathbf{M}$ be  a $d\times d$ matrix with integer entries and positive determinant.  For the case
$G=\mathbb{Z}^d$ and $M=\{\mathbf{M} n: n\in \mathbb{Z}^d\}$, we could take as transversal $\mathcal{L}=-\mathcal{N}(\mathbf{M})$  where  
$\mathcal{N}(\mathbf{M}):= \mathbf{M}[0,1)^n \cap \mathbb{Z}^d$ which has $\det \mathbf{M}$ elements (see \cite{vaid:93}). We could also take 
$\mathcal{L}=\mathcal{N}(\mathbf{M}^\top)$ (see \cite{vetterli:92}), or even other possibilities (see \cite{viscito:91}). In the following corollary we write some of the results of Section \ref{section3} in terms of the $K\times \det \mathbf{M}$ and $\det \mathbf{M}\times K$ {\em polyphase matrices} usually used in this context \cite{vetterli:92,vaid:93,viscito:91}:
\[
\mathbf{E}(z)=\Big[\sum_{n\in \mathbb{Z}^d} h_k(\mathbf{M}n-\ell)z^{-n}\Big]_{k\in \mathcal{K},\, \ell\in \mathcal{L}},\quad
\mathbf{R}(z)=\Big[ \sum_{n\in \mathbb{Z}^d} g_k(\mathbf{M}n+\ell ) z^{-n}\Big]_{ \ell\in \mathcal{L} ,k\in \mathcal{K}},\quad z\in \mathbb{T}^d.
\]

\medskip

\begin{cor}\label{cotasZd}
Under the above circumstances, consider the filter bank described in Fig.~\ref{f1}. Let $\lambda_{\min}(z)$  and $\lambda_{\max}(z)$ be  the smallest and the largest eigenvalue of the $\det \mathbf{M}\times \det \mathbf{M}$ matrix $\mathbf{E}^*(z)\mathbf{E}(z)$. Then, the sequence
$\big\{T_m f_k\big\}_{k\in\mathcal{K},\,m\in M}$ is a frame for $\ell^2(\mathbb{Z}^d)$ if and only if Rank $\mathbf{E}(z)=\det \mathbf{M}$ for all $z \in \mathbb{T}^d$.
In this case, the optimal frame bounds are
\[
A= \min_{z \in \mathbb{T}^d} \big[\lambda_{\min}(z)\big]\quad \text{and}\quad B= \max_{z \in \mathbb{T}^d} \big[\lambda_{\max}(z)\big]. 
\]
The sequences $\big\{T_m f_k\big\}_{k\in\mathcal{K},\,m\in M}$ and $\big\{T_m g_k\big\}_{k\in\mathcal{K},\,m\in M}$ are dual frames if and only if  $\mathbf{R}(z)\mathbf{E}(z)=\mathbf{I}_{\det \mathbf{M}}$ for all $z\in \mathbb{T}^d$. The sequence $\big\{T_m f_k\big\}_{k\in\mathcal{K},\,m\in M}$ is a tight frame if and only if $\mathbf{E}^*(z)\mathbf{E}(z)=A \,\mathbf{I}_{\det \mathbf{M}}$, $z\in \mathbb{T}^d$.  
Whenever $\det \mathbf{M}=K$, the sequence
$\big\{T_m f_k\big\}_{k\in\mathcal{K},\,m\in M}$ is a Riesz basis for $\ell^2(\mathbb{Z}^d)$ if and only if  $\det \mathbf{E}(z)\neq 0$ for all $z \in \mathbb{Z}^d$. In this case, the  optimal Riesz bounds are $A$ and $B$.
\end{cor}
\begin{proof}
For a matrix with integer entries $\mathbf{A}$  we define  $z^\mathbf{A}$ as the vector whose
$k$-component is $z_1^{\mathbf{A}_{1,k}}z_2^{\mathbf{A}_{2,k}}\ldots z_d^{\mathbf{A}_{d,k}}$. It can be verified that $[z^\mathbf{A}]^\mathbf{B}=z^{\mathbf{A}\mathbf{B}}$ (see \cite[pp. 581-582]{vaid:93}). Then
\[
\begin{split}
H_{k,\ell}(z + M^\perp)&=
\sum_{m\in M} h_k(m-\ell)z^{-m}=
\sum_{n\in\mathbb{Z}^d} h_k(\mathbf{M}n-\ell)z^{-\mathbf{M}n}\\&=
\sum_{n\in\mathbb{Z}^d} h_k(\mathbf{M}n-\ell)[z^{\mathbf{M}}]^{-n}=E_{k,\ell}(z^{\mathbf{M}}).
\end{split}
\]
 ($z + M^\perp$ denotes an element of $\mathbb{T}^d/M^\perp$) and analogously $G_{\ell,k}(z+ M^\perp)=R_{\ell,k}(z^{\mathbf{M}})$. Then
\[
\mathbf{H}(z + M^\perp)=\mathbf{E}(z^\mathbf{M}),\quad \mathbf{G}(z+ M^\perp)=\mathbf{R}(z^\mathbf{M}).
\]
 Besides, 
for any $z\in \mathbb{T}^d$ there exists $s\in \mathbb{T}^d$ such $s^\mathbf{M}=z$. 
Indeed,
there exists $r\in \mathbb{T}^d$ such that $r_{j}^{\det \mathbf{M}}=z_j$ and then 
$
[r^{\text{adj} \mathbf{M}}] ^{\mathbf{M}}= r^{(\text{adj} \mathbf{M}) \mathbf{M}}= r^{\mathbf{I}\det\mathbf{M}}=z
$.
By using these two facts, the corollary is a consequence of Theorems \ref{dual frames}, \ref{cotas} and \ref{maximally decimated} and Corollary \ref{tight}.
\end{proof}

\medskip

This corollary generalizes, to the multidimensional case, the results obtained in \cite{bol:98} and \cite{vetterli:98} for the unidimensional case.

\subsection{The case $G=\mathbb{Z}_s$ and $M=L\, \mathbb{Z}_s$}

Assume that $s=LN$, with $L, N\in \mathbb{N}$. Whenever $G=\mathbb{Z}_s$ and $ M=L\mathbb{Z}_s$
we  could take $\mathcal{L}=\{0,-1,\ldots, -(L-1)\}$ (mod $s$) (see \cite{vaid:97,vaid:99}). In the following corollary we write the results in terms of the $K\times L$ and $L\times K$ {\em polyphase matrices} defined in \cite{vaid:97,vaid:99}:
\[
\mathbf{E}(n)=\Big[\sum_{m=0}^{N-1} h_k(Lm-\ell)W_N^{-mn}\Big]_{k\in \mathcal{K},\, \ell\in \mathcal{L}},\quad
\mathbf{R}(n)=\Big[\sum_{m=0}^{N-1} g_k(Lm+\ell)W_N^{-mn}\Big]_{ \ell\in \mathcal{L} ,k\in \mathcal{K}},\quad n \in \mathbb{Z}_N.
\]
Note that the $N$-point DFT appears since $\widehat{M}\cong \mathbb{Z}_s/M^\perp \cong \mathbb{Z}_s/(N\mathbb{Z}_s)\cong \mathbb{Z}_N$.

\medskip
\begin{cor}\label{cotasZs}
Under the above circumstances, consider the filter bank described in Fig.~\ref{f1}. Let $\lambda_{\min}(n)$  and $\lambda_{\max}(n)$ be  the smallest and the largest eigenvalue of  the $L\times L$ matrix $\mathbf{E}^*(n)\mathbf{E}(n)$. The sequence
$\big\{T_m f_k\big\}_{k\in\mathcal{K},\,m\in M}$ is a frame for $\ell^2(\mathbb{Z}_s)$ if and only if Rank $\mathbf{E}(n)=L$ for all $n \in \mathbb{Z}_N$.
In this case, the optimal frame bounds are
\[A= \min_{n \in \mathbb{Z}_N} \big[\lambda_{\min}(n)\big]\quad \text{and} \quad B= \max_{n \in \mathbb{Z}_N} \big[\lambda_{\max}(n)\big]. 
\]
It is tight frame if and only if $\mathbf{E}^*(n)\mathbf{E}(n)=A \mathbf{I}_L$  for all $n\in \mathbb{Z}_N$.  The sequences 
$\big\{T_m f_k\big\}_{k\in\mathcal{K},\,m\in M}$ and $\big\{T_m g_k\big\}_{k\in\mathcal{K},\,m\in M}$ are dual frames if and only if  $\mathbf{R}(n)\mathbf{E}(n)=\mathbf{I}_L$ for all $n\in \mathbb{Z}_N$. Whenever $L=K$, the sequence
$\big\{T_m f_k\big\}_{k\in\mathcal{K},\,m\in M}$ is a Riesz basis for $\ell^2(\mathbb{Z}_s)$ if and only if  $\det \mathbf{E}(n)\neq 0$ for all $n \in \mathbb{Z}_N$. In this case, the  optimal Riesz bounds are $A$ and $B$.
\end{cor}
\begin{proof}
We have 
$\widehat{M}\cong \widehat{G}/M^\perp \cong \mathbb{Z}_{s}/M^\perp $ with
 $\langle Lm,n+M^\perp \rangle=W_s^{mLn}=W_N^{mn}$. Then
\[
H_{k,\ell}(n+M^\perp)=\sum_{m=0}^{N-1} h_k(Lm-\ell)\overline{\langle Lm,n \rangle}=
\sum_{m=0}^{N-1} h_k(Lm-\ell)W_N^{-mn}=E_{k,\ell}(n)
\]
and analogously $G_{\ell,k}(n)=R_{\ell,k}(n)$. Hence,the corollary is a consequence of Theorems \ref{dual frames}, \ref{cotas} and \ref{maximally decimated} and Corollary \ref{tight}, having in mind that $\mathbf{E}(n)$ and $\mathbf{R}(n)$ are $N$-periodic. 
\end{proof}

\medskip

Some of these results can be found in \cite{chai:07,chebira:11,johnson:08,fickus:13}.

\subsection{The case $G=\mathbb{Z}^d \times \mathbb{Z}_s$ and $M= \mathbf{M}\,\mathbb{Z}^d \times L\, \mathbb{Z}_s$}

Consider now the tensor product of the two previous examples, i.e.,
$G=\mathbb{Z}^d \times \mathbb{Z}_s$ and $M= \mathbf{M}\mathbb{Z}^d \times L\mathbb{Z}_s$,
where $\mathbf{M}$ is a matrix with integer entries, $\det \mathbf{M}>0$ and $s=L N$. 
We could take $\mathcal{L}=\mathcal{N}(\mathbf{M})\times \{0,1,\ldots, (L-1)\}$.  
Set the $K\times L$ and $L\times K$ matrices
\[
\begin{split}
&\mathbf{E}(z,n)=\Big[\sum_{m=0}^{N-1}\sum_{u\in \mathbb{Z}^d} h_k\big([\mathbf{M}u,Lm]-\ell)z^{-u}W_N^{-mn}\Big]_{k\in \mathcal{K},\, \ell\in \mathcal{L}}
\\&\mathbf{R}(z,n)=\Big[\sum_{m=0}^{N-1}\sum_{u\in \mathbb{Z}^d} g_k\big([\mathbf{M}u,Lm]+\ell)z^{-u}W_N^{-mn}\Big]_{\ell\in \mathcal{L}, \,k\in \mathcal{K}}
\end{split}
\]

\medskip
\begin{cor}\label{cotasZdZs} Under the above circumstances, the filter bank described in Fig.~\ref{f1} satisfies the perfect reconstruction property if and only if $\mathbf{R}(z,n)\mathbf{E}(z,n)=\mathbf{I}_{L\det \mathbf{M}}$  for all $z\in \mathbb{T}^d$ and $n\in \mathbb{Z}_N$.
Let $\lambda_{\min}(z,n)$  and $\lambda_{\max}(z,n)$ be  the smallest and the largest eigenvalue of the $L\det \mathbf{M} \times L\det \mathbf{M}$ matrix  $\mathbf{E}^*(z,n)\mathbf{E}(z,n)$. The sequence
$\big\{T_m f_k\big\}_{k\in\mathcal{K},\,m\in M}$ is a frame for $\ell^2(\mathbb{Z}^d \times \mathbb{Z}_s)$ if and only if Rank $\mathbf{E}(z,n)=L\det \mathbf{M}$ for all $z \in \mathbb{T}^d$ and $n\in \mathbb{Z}_N$.
In this case, the optimal frame bounds are
\[A= \min_{z \in \mathbb{T}^d,n\in \mathbb{Z}_N}\big[ \lambda_{\min}(z,n)\big]\quad \text{and} \quad B= \max_{z \in \mathbb{T}^d,n\in \mathbb{Z}_N}\big[ \lambda_{\max}(z,n)\big]. 
\]
 Whenever $K=L\det \mathbf{M}$, the sequence $\big\{T_m f_k\big\}_{k\in\mathcal{K},\,m\in M}$ is a Riesz basis for $L^2(\mathbb{Z}^d \times \mathbb{Z}_s)$ if and only if  $\det \mathbf{E}(z,n)\neq 0$ for all $z \in \mathbb{Z}^d$ and $n\in \mathbb{Z}_N$. In this case, the  optimal Riesz bounds are $A$ and $B$.
\end{cor}
\begin{proof}
We have
\[
\begin{split}
H_{k,\ell}\big((z,n) + M^\perp\big)&=
\sum_{u\in\mathbb{Z}^d} \sum_{m=0}^{N-1}h_k\big((\mathbf{M}u,mL)-\ell\big) \overline{\langle (\mathbf{M}u,mL), (z,n)\rangle}\\&=
\sum_{u\in\mathbb{Z}^d} \sum_{m=0}^{N-1}h_k\big((\mathbf{M}u,mL)-\ell\big)z^{-\mathbf{M}u} W_s^{-mLn}\\&=
\sum_{u\in\mathbb{Z}^d} \sum_{m=0}^{N-1} h_k\big((\mathbf{M}u,mL)-\ell\big)[z^{\mathbf{M}}]^{-u} W_N^{-mn}=E_{k,\ell}(z^{\mathbf{M}},n)
\end{split}
\]
and analogously $G_{\ell,k}\big((z,n)+ M^\perp\big)=R_{\ell,k}(z^{\mathbf{M}},n)$.  Besides (it was proved in previous proof)
for any $z\in \mathbb{T}^d$ there exist $s\in \mathbb{T}^d$ such $s^\mathbf{M}=z$. 
Thus, having in mind the $N$-periodicity, the corollary is a consequence of Theorems \ref{PR}, \ref{cotas} and \ref{maximally decimated}.
\end{proof}

\subsection{The case $G=\mathbb{Z}_{2P} \times\mathbb{Z}_{2Q}$ and $M$ the Quincunx}
Given $P, Q\in \mathbb{N}$, the Quincunx $M$ consists of  the elements $(n,m)$ in $\mathbb{Z}_{2P} \times \mathbb{Z}_{2Q}$ such that $n$ and $m$ are both even or both odd; it  is a subgroup of $\mathbb{Z}_{2P} \times \mathbb{Z}_{2Q}$. In this case we could take $\mathcal{L}=\{(0,0),(1,0)\}$. Consider the $[P,Q]$-Points DFT transform
  \[
 \big[ \text{DFT} \, x \big] (n,m)=\sum_{u=0}^{P-1}
  \sum_{v=0}^{Q-1} x(u,v)W_P^{-un}\,W_Q^{-vm}\,,
  \]
and the transform
\[
[\Lambda x](n,m)= \big[ \text{DFT}\, x_0\big] (n,m)+W_{2P}^{-n}\,W_{2Q}^{-m}\big[ \text{DFT}\,  x_1\big] (n,m),
\]  
where 
 $x_0(n,m)=x(2n,2m)$ and $x_1(n,m)=x(2n+1,2m+1)$. Respectively set the $K\times 2$ and $2\times K$ matrices
  \[
  \mathbf{E}(n,m)=\Big[\Lambda h_{k,\ell}(n,m)\Big]_{k\in \mathcal{K},\, \ell\in \mathcal{L}}\quad \text{and} \quad
  \mathbf{R}(n,m)=\Big[\Lambda g_{\ell,k}(n,m)\Big]_{ \ell\in \mathcal{L}, k\in \mathcal{K}}\,.
  \]  

\begin{cor}\label{quincunx} Under the above circumstances, the filter bank described in Fig.~\ref{f1} has the perfect reconstruction property if and only if $\mathbf{R}(n,m)\,\mathbf{E}(n,m)=\mathbf{I}_{2}$ for all $(n,m)\in \mathbb{Z}_{2P} \times \mathbb{Z}_Q$.
\end{cor}
\begin{proof}
We have 
$\widehat{M} \cong \widehat{G}/M^\perp \cong (\mathbb{Z}_{2P}\times \mathbb{Z}_{2Q})/M^\perp$ with
\[
\begin{split}
&\big\langle (2u,2v), (n,m)+M^\perp \big\rangle= W_{2P}^{2un}\, W_{2Q}^{2vm}= W_{P}^{un} \,W_{Q}^{vm}\\
&\big\langle (2u+1,2v+1), (n,m)+M^\perp \big\rangle= W_{2P}^{(2u+1)n} \,W_{2Q}^{(2v+1)m}= W_{2P}^{n}\, W_{2Q}^{m} \,W_{P}^{un}\, W_{Q}^{vm}.
\end{split}
\]
Then, the $M$-Fourier transform of a function $h$ is given by
\[
\begin{split}
&H\big((n,m)+M^\perp\big)= 
\sum_{u=0}^{P-1}\sum_{v=0}^{Q-1}\big[ h_0(u,v) W_{P}^{-un}\, W_{Q}^{-vm}+ h_1(u,v) W_{P}^{-un}\, W_{Q}^{-vm} \,W_{2P}^{-n} \,W_{2Q}^{-m}\big].
\end{split}
\]
Hence, from Theorem \ref{PR},  the filter bank satisfies the perfect reconstruction property if and only if $\mathbf{R}(n,m)\,\mathbf{E}(n,m)=\mathbf{I}_2$ for all $(n,m)\in \mathbb{Z}_{2P} \times \mathbb{Z}_{2Q}$. Since $\Lambda(n+P,m+Q)=\Lambda(n,m)$ and
$\Lambda(n,m+Q)=\Lambda(n+P,m)$, it suffices to consider $(n,m)\in \mathbb{Z}_{2P} \times \mathbb{Z}_{Q}$.
\end{proof}

Note that $\widehat{M}\cong (\mathbb{Z}_{2P}\times \mathbb{Z}_{2Q})/M^\perp \cong (\mathbb{Z}_{2P}\times \mathbb{Z}_{2Q})/\{(0, 0), (P, Q)\} \cong \mathbb{Z}_{2P}\times \mathbb{Z}_{Q}$.

\bigskip

\noindent{\bf Acknowledgments:} This work has been supported by the grant MTM2014-54692-P from the Spanish {\em Ministerio de Econom\'{\i}a y Competitividad (MINECO)}


\begin{thebibliography}{10}


\bibitem{barbieri:15b}
D. Barbieri, E. Hern\'andez and V. Paternostro.
\newblock The Zak transform and the structure of spaces invariant by the action of an LCA group,
\newblock {\em J. Funct. Anal. }, 269, 1327--1358, 2015.
 

\bibitem{barbieri:15}
D. Barbieri, E. Hern\'andez, J. Parcet.
\newblock Riesz and frame systems generated by unitary actions of discrete groups.
\newblock {\em Appl. Comput. Harmon. Anal.}, 39(3): 369--399, 2015.
 
\bibitem{Behmard:99}
H. Behmard.
\newblock {\em Nonperiodic Sampling Theorems and Filter Banks}.
\newblock  {Ph.D.
thesis, Dept. of Mathematics, Oregon State University, Corvallis}, OR
97331, 1999.

\bibitem{bol:98}
H. B\"{o}lcskei, F. Hlawatsch, and H. G. Feichtinger.
\newblock {\em Frame-theoretic analysis of oversampled filter banks}.
\newblock {  IEEE Trans. Signal Process.}, 46(12):3256--3268, 1998.

\bibitem{bol:97}
H. B\"{o}lcskei and F. Hlawatsch.
\newblock {\em Discrete Zak transforms, polyphase
transforms, and applications}.
\newblock {  IEEE Trans. Signal Process.}, 45(12):851--866, 1997.

\bibitem{cabrelli:10}
C. Cabrelli and V. Paternostro.
\newblock Shift-invariant spaces on LCA groups.
\newblock {\em J. Funct. Anal.}, 258(6): 2034--2059, 2010.

\bibitem{cabrelli:12}
C. Cabrelli and V. Paternostro.
\newblock Shift-modulation invariant spaces on LCA groups.
\newblock {\em Studia Math.}, 211:1--19, 2012.

\bibitem{chai:07}
 L. Chai, J. Zhang, C. Zhang, and E. Mosca. 
\newblock {\em Frame-theory-based analysis and design of oversampled filter banks: Direct
computational method}.
\newblock { IEEE Trans. Signal Process.}, 55(2):507--519, 2007.

\bibitem{chebira:11}
A. Chebira, M. Fickus, and  D.G. Mixon.
\newblock {\em Filter bank fusion frames}.
\newblock {IEEE Trans. Signal Process.}, 59:953--963, 2011.

\bibitem{ole:03}
O.~Christensen.
\newblock {\em An {I}ntroduction to {F}rames and {R}iesz {B}ases}.
\newblock Birkh{\"a}user, Boston, 2003.

\bibitem{ole:15}
O.~Christensen and S.~S.~Goh.
\newblock {Fourier-like frames on locally compact abelian groups}.
\newblock {\em J. Approx. Theor.}, 192:82--101, 2015.

\bibitem{vetterli:98}
Z. Cvetkovi\'c and M. Vetterli.
\newblock {\em Oversampled filter banks}.
\newblock { IEEE Trans. Signal Process.}, 46:1245--1255, 1998.

\bibitem{dodson:07}
M.~M.~Dodson.
\newblock Groups and the Sampling Theorem.
\newblock {\em Sampl. Theory Signal Image Process.}, 6(1):1--27, 2007.

\bibitem{fickus:13}
M. Fickus, M.L. Massar, and D.G. Mixon.
\newblock {\em Finite Frames and Filter Banks}.
\newblock { In Finite
Frames: Theory and Applications, P.G. Casazza and G. Kutyniok (Eds.). Birkhauser}, 337--380, 2013.

\bibitem{folland:95}
G.~B. Folland.
\newblock {\em A Course in Abstract Harmonic Analysis}.
\newblock CRC Press, Boca Raton FL, 1995.

\bibitem{garcia:16}
A.~G.~Garc\'{\i}a, M.~A.~Hern\'andez-Medina, and G.~P\'erez-Villal\'on.
\newblock Sampling in unitary invariant subspaces associated to LCA groups.
\newblock ArXiv:1605.04127v1[math.FA], 2016.

\bibitem{hernandez:10}
E.~Hern\'andez, H.~Sikic, G.~Weiss and E.~Wilson.
\newblock Cyclic subspaces for unitary representations of LCA groups; generalized Zak transform.
\newblock {\em Colloq. Math.}, 118(1):313--332, 2010.

\bibitem{Hewitt:70}
E. Hewitt and K.A. Ross
\newblock {\em Abstract Harmonic Analysis}, Vol 1 and 2.
\newblock  {Springer Verlag Berlin}, 1970.

\bibitem{johnson:08}
B.D. Johnson, and K. Okoudjou.
\newblock {\em  Frame potential and finite abelian groups}.
\newblock {Contemp. Math.}, 464:137--148, 2008.

\bibitem{Kalker:96}
T. Kalker and I. Shah.
\newblock {\em A Group Theoretic Approach
to Multidimensional Filter Banks: Theory and Applications}.
\newblock { IEEE Trans. Signal Process.}, 44(6):1392--1405, 1996.

\bibitem{kaniuth:98}
E. Kaniuth and G. Kutyniok.
\newblock {\em Zeros of the Zak transform on locally compact abelian groups}.
\newblock {Proc. Amer. Math. Soc.},  126: 3561--3569, 1998.

\bibitem{Klappenecker:98}
A. Klappenecker and M. Holschneider.
\newblock {\em A unified view on filter banks}.
\newblock {  Wavelet Applications in
Signal and Image Processing VI, A. F. Laine, M. A. Unser, and A. Aldroubi, eds}, 3458:2--13, SPIE, 1998.

 \bibitem{vetterli:92}
J. Kovacevi\'c and M. Vetterli.
\newblock {\em Nonseparable multidimensional perfect reconstruction filter banks and wavelets bases for $\mathbb{R}^n$}.
\newblock { IEEE Trans. Inform. Theory}, 38:533--555, 1992.

\bibitem{mallat:09}
S.~Mallat.
\newblock{\em A wavelet Tour of Signal Processing}.
\newblock Academic Prees, Burlington MA, 2009.

\bibitem{rudin:90}
W.~Rudin.
\newblock {\em Fourier Analysis on Groups}.
\newblock Wiley, Wiley Classics Library, New York, 1990.

\bibitem{sinha:10}
V.P. Sinha.
\newblock {\em Symmetries and Groups in Signal Processing}.
\newblock Springer, New York, 2010.

\bibitem{strang:96}
G.~Strang and T.~Nguyen.
\newblock {\em Wavelets and Filter Banks}.
\newblock Wellesley-Cambridge Press, MA 1996.

\bibitem{vaid:87}
P. P. Vaidyanathan.
\newblock {\em Theory and design fo M-channel maximally decimated quadrature mirror filters with arbitrary M, having perfect reconstruction property}.
\newblock { IEEE Trans. on Acoustics, Speech and Signal Processing.}, ASSP-35: 476--492, 1987.

\bibitem{vaid:93}
P. P. Vaidyanathan.
\newblock {\em Multirate Systems and Filter Banks}.
\newblock Prentice Hall,
1993.

\bibitem{vaid:97}
P. P. Vaidyanathan and A. Kirac.
\newblock {\em Theory of cyclic filter banks}.
\newblock {Proc. IEEE Int. Conf. Acoust., Speech, Signal Process.},  pp. 2449--2452, 1997.


\bibitem{vaid:99}
P. P. Vaidyanathan and A. Kirac.
\newblock {\em Cyclic LTI Systems in Digital Signal
Processing}.
\newblock { IEEE Trans. Signal Process.}, 47(2):433--447, 1999.

\bibitem{vetterli:86}
M.~ Vetterli 
\newblock {\em Filter banks allowing for perfect reconstruction}.
\newblock { Signal Processing.} 10:219--244, 1986.

\bibitem{vetterli:95}
M.~ Vetterli and J.~Kova$\check{c}$evi\'c
\newblock {\em Wavelets and Subband Coding}.
\newblock Prentice Hall, 1995.

\bibitem{viscito:91}
E. Viscito and J. P. Allebach.
\newblock {\em The analysis and design of multidimensional
FIR perfect reconstruction filter banks for arbitrary sampling lattices.}
\newblock { IEEE Trans. Circ. and Syst.}, 38(1):29--42, 1991.


\bibitem{zedek:65}
M. Zedek.
\newblock {\em Zeroes of Linear Combinations of Polynomials}.
\newblock {Proc. Amer. Math. Soc.}, 16:78--84, 1965.

\end{thebibliography}
\end{document}